\documentclass[11pt]{article}[12pt]
\usepackage[T2A]{fontenc}
\usepackage[utf8]{inputenc}
\usepackage[english]{babel}
\usepackage{amsmath}
\usepackage{amssymb}
\usepackage{amsthm}
\usepackage{graphicx}
\usepackage{epstopdf}
\usepackage[affil-it]{authblk}
\usepackage{upgreek}
\usepackage{diagbox}
\usepackage{multirow}
\usepackage{appendix}
\usepackage{hyperref}
\hypersetup{
colorlinks=true,
linkcolor=red,
citecolor=blue,
urlcolor=blue
}

\usepackage[a4paper, total={6in, 9.6in}]{geometry}

\def\R{\mathbb{R}}
\def\C{\mathbb{C}}
\newcommand{\pd}[2]{\dfrac{\partial #1}{\partial #2}}

\newcommand{\der}[2]{\dfrac{\mathrm{d} #1}{\mathrm{d} #2}}
\newcommand{\bra}[1]{\langle #1 |}
\newcommand{\ket}[1]{| #1 \rangle}
\newcommand{\scalarproduct}[2]{\langle #1 | #2 \rangle}

\def\I{\mathbb{I}}      %identity

\def\br{\mathbf{r}}     %bold r for vector
\def\bq{\mathbf{q}}     %bold q for vector
\def\Tr{{\rm Tr}}       %trace symbol Tr

\newtheorem{proposition}{Proposition}

\title{\textbf{Quantum сontrol landscapes for generation of ${H}$ and ${T}$ gates in an open qubit with  both coherent and environmental drive}}

\date{}

\author[1,2,*]{Vadim N. Petruhanov}
\author[1,2,**]{Alexander N. Pechen}

\affil[1]{\it \normalsize Department of Mathematical Methods for Quantum Technologies,\par
Steklov Mathematical Institute of Russian Academy of Sciences,\par
8~Gubkina str., Moscow, 119991, Russia, }
\affil[2]{\it Quantum Engineering Research and Education Center,\par
University of Science and Technology MISIS,\par
6~Leninskiy prospekt, Moscow, 119991, Russia;}
\affil[*]{vadim.petrukhanov@gmail.com, \href{http://www.mathnet.ru/eng/person176798}{mathnet.ru/eng/person176798}}
\affil[**]{apechen@gmail.com, \href{http://www.mathnet.ru/eng/person17991}{mathnet.ru/eng/person17991}}

\begin{document}

\maketitle

\begin{abstract} An important problem in quantum computation is the generation of single-qubit quantum gates such as Hadamard ($H$) and $\pi/8$ ($T$) gates, which are components of a  universal set of gates. Qubits in experimental realizations of quantum computing devices are interacting with their environment. While the environment is often considered as an obstacle leading to a decrease in the gate fidelity, in some cases, it can be used as a resource. Here, we consider the problem of the optimal generation of $H$ and $T$ gates using coherent control and the environment as a resource acting on the qubit via incoherent control. For this problem, we studied the quantum control landscape, which represents the behavior of the infidelity as a functional of the controls. We considered three landscapes, with infidelities defined by steering between two, three (via Goerz--Reich--Koch approach), and four matrices in the qubit Hilbert space. We observed that, for the $H$ gate, which is a Clifford gate, for all three infidelities, the distributions of minimal values obtained with a gradient search have a simple form with just one peak. However, for the $T$ gate, which is a non-Clifford gate, the situation is surprisingly different---this distribution for the infidelity defined by two matrices also has one peak, whereas distributions for the infidelities defined by three and four matrices have two peaks, which might indicate the possible existence of two isolated minima in the control landscape. It is important that, among these three infidelities, only those defined with three and four matrices guarantee the closeness of the generated gate to a target and can be used as a good measure of closeness. We studied sets of optimized solutions for the most general and previously unexplored case of coherent and incoherent controls acting together and discovered that they form sub-manifolds in the control space, and unexpectedly, in some cases, two isolated sub-manifolds.

\end{abstract}

\medskip

\noindent{{\bf Keywords:} incoherent control; control by environment; open quantum system; qubit; \\quantum gate generation; gradient method}

\medskip

\section{Introduction}\label{Sec1}

Quantum computation is an actively developing field within the general area of quantum technologies~\cite{Schleich2016,Acin2018}. An important problem in quantum computation is the generation of single-qubit quantum gates such as the Hadamard ($H$) gate and $\pi/8$ ($T$) gate, which are defined by unitary matrices 
\begin{equation}
H= \frac{1}{\sqrt{2}}\begin{pmatrix}
1 & 1\\
1 & -1
\end{pmatrix},\qquad T = \begin{pmatrix}
1 & 0\\
0 & e^{i\pi/4}
\end{pmatrix}.%added period '.'
\label{H_T_gates}
\end{equation}
These gates together with the two-qubit C-NOT gate, which, along with the $H$ gate, is a Clifford gate and can be efficiently simulated classically according to the Gottesman--Knill theorem~\cite{Gottesman_1998}, form a universal quantum gate set sufficient for universal quantum computation~\cite{NielsenChuangBook}. Therefore, the practical realization of $H$ and $T$ gates is of high importance for quantum computation.

Qubits in experimental realizations of quantum computation devices are interacting with their environment. Therefore, they are open quantum systems whose dynamics are non-unitary and is described instead of the Sch\"odinger equation by various master equations~\cite{BreuerPetruccioneBook}. The environment is often considered as an obstacle for manipulating the qubits. However, in some cases, it can serve as a useful resource. It can be used, for example, for quantum computing with mixed states and non-unitary quantum gates, which was proposed in~\cite{Aharonov_1998,Tarasov_2002}. Engineered environments were suggested in a more general context for improving quantum computation~\cite{Cirac2009}, cooling translational motion~\cite{Calarco2011}, preparing entangled states~\cite{Diehl_nature_2008,Weimer_nature_2010}, inducing multi-particle entanglement dynamics~\cite{Barreiro_nature_2010}, making robust quantum memories~\cite{Pastawski_2011}, dissipative control of a quantum spin chain~\cite{Morigi_Eschner_Cormick_Lin_Leibfried_Wineland_2015}, the dissipative preparation of many-body quantum states in a superconducting qutrit array~\cite{Wang_Snizhko_Romito_Gefen_Murch_2023}, and inducing stationary quantum memory effects~\cite{Sun_Luoma_Liu_Piilo_Li_Guo_2023}.
The control of one- and two-qubit systems under dissipative conditions was studied~\cite{Kallush_Dann_Kosloff_2022}. Various results in the quantum control of open quantum systems are discussed in the review~\cite{KochEPJQuantumTechnol2022}.

An early approach for using the environment and dissipation, generally with time-dependent decoherence rates and time-dependent master equations, for controlling quantum systems was developed in 2006 and was called \emph{incoherent control}~\cite{Pechen_Rabitz_2006}. In this approach, spectral density of the environment is used as the control function (which is generally time-dependent) to manipulate quantum systems.

This spectral density represents the distribution of the particles of the environment in their momenta and internal degrees of freedom. For example, for the environment formed by incoherent (and not necessarily thermal) photons, it describes the distribution of photons in their momenta and polarization. This distribution apparently affects the reduced dynamics of a quantum system immersed in this environment and enters in the master equation for the reduced density matrix. This is well known for two types of the environment---weakly interacting with the system (weak coupling limit)~\cite{DaviesBook,AccardiLuVolovichBook,SpohnLebowitzReview}, or interacting with the system strongly but rarely (low density limit)~\cite{Dumcke1985,PechenJMP2004,VacchiniPR2009}, which were both considered in~\cite{Pechen_Rabitz_2006}. Master equations derived beyond secular approximation and for ultrastrong-coupling and the strong-decoherence limits~\cite{Trushechkin_2021,Trushechkin_2022} may also be of interest for an investigation.

If the spectral density depends on time, then the master equation for the reduced density matrix will also generally be time-dependent. For incoherent control, the master equations for the description of the dynamics of a quantum system controlled by the spectral density of the environment of the following form was proposed~\cite{Pechen_Rabitz_2006} as
\begin{equation}\label{eq:ME}
\frac{d\rho_t}{dt}=-i[H_c(t),\rho_t]+\sum\limits_k \gamma_k(t) {\cal D}_k.
\end{equation}
where $H_c(t)$ is some controlled Hamiltonian, $\gamma_k(t)$ are the generally controlled \emph{time-dependent decoherence rates}, and ${\cal D}_k$  are some dissipators. In~\cite{Pechen_Rabitz_2006}, beyond the general approach, two particular physical forms of the GKSL dissipators ${\cal D}_k$ were considered based on the derived in 1970th and 1980th master equations for the weak coupling and low density limits. A more detailed description of the concept of incoherent control is also provided in Section~3 of~\cite{PetruhanovPechenJPA2023}. 

A natural question refers to how useful such incoherent control can be. In this regard, in~\cite{Pechen2011}, it was shown that incoherent control when combined together with coherent control using lasers can be used for the approximate generation of arbitrary mixed density matrices for generic (i.e., almost for all) quantum systems. It is important that this was shown within the physical class of ${\cal D}_k$, which describe the reduced dynamics of a quantum system weakly interacting with the environment (e.g., with incoherent photons or phonons). Coherent control alone cannot achieve the control goal of preparing arbitrary mixed quantum states for generic quantum systems. Moreover, the control scheme proposed in~\cite{Pechen2011} also allows the implementation of \emph{all-to-one} or \emph{universally optimal Kraus maps}, i.e., Kraus maps, which steer all initial density matrices into the same one density matrix and were introduced and studied for quantum control in~\cite{Wu_2007_5681}. In this work, these maps were shown to be important for quantum control, since they are inherently robust to variations in the system's initial state. The generation of some maps of this class was also discussed recently in~\cite{Kallush_Dann_Kosloff_2022}. There was a non-trivial problem of the experimental realization of such all-to-one Kraus maps, which was recently solved for an open single qubit~\cite{Zhang_Saripalli_Leamer_Glasser_Bondar_2022}. Other various experimental works have been conducted using control by incoherent photons or the modification of relaxation rates. In particular, shaped incoherent light was used for the optimization of up-conversion hues in phosphors~\cite{LaforgeJCP2018}. The experimental control of the system--reservoir interaction and hydrogen spin relaxation rates was achieved~\cite{Pires2023}. The low density limit case with collisional-type decoherence is less studied in this context, although it describes such important examples as a test particle in a quantum gas~\cite{VacchiniPRE2001} or the quantum linear Boltzmann equation~\cite{VacchiniPR2009}. 

Recently, a surprising new result was obtained for a qubit interacting with the environment driven by coherent and incoherent controls~\cite{Lokutsievskiy_2021}. It was shown that inside of the Bloch ball, many states can be obtained exactly, except for states in some domain of size $\delta\approx\gamma/\omega$, where $\gamma$ is the decoherence rate and $\omega$ is the qubit transition frequency. Moreover, reachable sets of states for a qubit driven by coherent and incoherent controls were analytically described using geometric control theory. A reachable set of states is the set of states that can be obtained from a given state using all available coherent and incoherent controls and any time. A controllability analysis of the quantum systems immersed within an engineered environment was performed in~\cite{Grigoriu_Rabitz_Turinici_2013}.

The analysis of controllability and the description of reachable sets answers the question of what states can, in principle, be created from a given initial state for a particular quantum system. For any state in the reachable set, an optimal control that steers the initial state into this state exists. Thus, the analysis of reachable sets and controllability answers the question of the principle existence of optimal control for a given control problem, but does not say how to find such controls. As soon as the existence of an optimal control is established, the next question is how to find this optimal control. 

According to~\cite{Bondar2020}, there is no single algorithm that, for any given set of controls and any pair of initial and target states, answers whether the initial state can be transferred into the target state without using these controls. %Please ensure the meaning has been retained.  %we confirm
 Despite this negative result, for a particular class of problems, such an algorithm may exist. For optimizing only coherent control, various algorithms were applied or developed, including the genetic algorithm~\cite{Judson_Rabitz_1992}, the Krotov algorithm~\cite{Tannor_Kazakov_Orlov_1992}, the Hamilton--Jacobi--Bellman equations~\cite{Gough_Belavkin_Smolyanov_2005}, 
chopped random-basis quantum optimization (CRAB)~\cite{Caneva_Calarco_Montangero_2011}, the Maday--Turinici algorithm~\cite{Maday_Turinici_2003}, GRAPE~\cite{Khaneja2005,DeFouquieres_Schirmer_Glaser_Kuprov_2011,Lucarelli_2018,Goodwin_Vinding_2023}, quantum feedback control~\cite{Wiseman_Milburn_1993,Doherty_Habib_Jacobs_Mabuchi_Tan_2000,Lloyd_Viola_2001,VanHandel_Stockton_Mabuchi_2005,Gough_2012,Schirmer_Jonckheere_Langbein_2018}, monotonically convergent algorithms~\cite{Turinici_2003,Lapert_Tehini_Turinici_Sugny_2008}, quantum reinforcement learning~\cite{Dong_IEEE_2008} and quantum machine learning~\cite{Biamonte_Wittek_Pancotti_Rebentrost_Wiebe_Lloyd_2017}, deep reinforcement learning~\cite{Niu_Boixo_Smelyanskiy_Neven_2019}, the combined approach via the quantum optimal control suite (QuOCS)~\cite{QuOCS2022}, etc. For finding both coherent and incoherent controls, genetic evolutionary algorithms were initially used~\cite{Pechen_Rabitz_2006}. Recently, the speed gradient method~\cite{PechenBorisenokFradkov2022},
gradient projection methods~\cite{MorzhinPechenQIP2023}, the Krotov method, and stochastic free-gradient optimization methods~\cite{MorzhinPechenIrkutsk2023} were adapted. 

A particular class of such optimization methods is that of local optimization methods, an example of which is the gradient ascent pulse engineering (GRAPE) approach~\cite{Khaneja2005}. This approach was extended to open quantum systems driven by coherent and incoherent controls in~\cite{PetruhanovPechenJPA2023}, where, in addition to the general scheme for a qubit, a new exact analytical expression for gradients of various objectives for a single qubit were derived by solving a cubic equation via the Cardano method. It is important that the efficiency of gradient-based approaches depends on the existence or absence of traps---local but not global maxima for the maximization of the objective, or minima for the minimization of the objective. The problem of the analysis of quantum control landscapes was posed first in~\cite{Rabitz_Hsieh_Rosenthal_2004}. Many results have been obtained since then, including proof of the absence of traps for a single qubit~\cite{Pechen_Ilin_2012,Volkov_Morzhin_Pechen_2021}, establishing the presence of trapping behavior in quantum systems with various symmetries~\cite{PechenTannorPRL,SchirmerIDAQP2013,VolkovPechenUMN2023,ElovenkovaQR2023}, etc.

Motivated by this, here we considered the problem of generating single-qubit $H$ and $T$ gates for an open qubit using coherent and incoherent controls. The gradient-based approach for solving this problem was developed in~\cite{PetruhanovPhotonics2022}. In the present work, we studied quantum control landscapes for this problem, which determines whether the GRAPE approach can be efficiently applied or its applications should take into account possible traps. In the latter case, runs from different starting points were performed for more efficient optimization. We considered three different landscapes for this problem, corresponding to three different objective functional defined by two, three, and four matrices, correspondingly, in the two-dimensional qubit Hilbert space.
For each gate and each objective functional, we performed runs of the GRAPE approach from multiple random initial starting points and, for each starting point, we computed the best (minimal, since we considered infidelities) objective value attained by the GRAPE approach. Then, we built the resulting distributions of the best obtained values. The observed results are surprising. For the Hadamard gate, which is a Clifford gate,  the distributions of best objective values for the all three objective functionals have a simple form with just one peak. However, for the $T$ gate, which is a non-Clifford gate, the situation is different---the distribution for the objective functional defined by the two matrices also has one peak, whereas the distributions of the best obtained values for objective functionals defined by the three and the four matrices have two isolated peaks, which might indicate the possible existence of two isolated minima with different values of the infidelity. It is important that, among these three objective functionals, the smallness of the objective with two matrices does not guarantee the closeness of the generated process to the target gate. Only objectives defined with three and four matrices guarantee the closeness of the generated gate to a target and can be used as a good measure for the gate generation problem, and, exactly for these objectives, we observed two peaks. An important feature of quantum optimal control problems is that optimal solutions are not isolated points but multi-dimensional sub-manifolds in the control space, as was shown for the coherent control of finite-dimensional~\cite{RabitzPRA2006} and continuous-variable systems~\cite{LaroccaPRA2020}. We also studied manifolds of optimized solutions and found that, in the case of both coherent and incoherent controls, the optimized solutions also tend to form manifolds, but interestingly, in some cases, they form two isolated manifolds. 

The structure of the paper is the following. In Section~\ref{Sec2}, the master equation for a qubit driven by coherent and incoherent controls is provided. In Section~\ref{Sec3},  three objective functionals for the single-qubit gate generation problem are defined and discussed, based on two, three, and four matrices in the qubit Hilbert space, and an objective functional defined directly in terms of quantum channels. The basic extension of the gradient-based optimization method for this optimization problem is provided in Section~\ref{Sec4}. The numeric analysis of the control landscapes is outlined in Section~\ref{Sec5}. Section~\ref{Conclusions} summarizes the~results.

\section{Environment-Assisted Control of a Qubit}\label{Sec2}

We consider the master equation describing the evolution of an open two-level quantum system (qubit) driven by coherent and incoherent controls: 
\begin{equation}
\frac{d\rho}{dt}  = - i [H_0 + V u(t), \rho] + \gamma \mathcal{L}_{n(t)}(\rho),\quad \rho(0) = \rho_0,
\label{system_qubit}
\end{equation}
where $\rho(t)$ is a $2\times2$ density matrix representing a state of the system at time $t\in[0,T]$, and $H_0$ and $V$ are the free Hamiltonian and the interaction Hamiltonian, respectively:
\begin{equation*}
H_0 = \omega \begin{pmatrix}
0 & 0 \\
0 & 1
\end{pmatrix},\qquad V = \mu\sigma_x
= \mu\begin{pmatrix}
0 & 1 \\
1 & 0
\end{pmatrix},
\end{equation*}
where $\omega$ is the qubit frequency, $\mu>0$ is the dipole moment, $\gamma$ is the decoherence rate coefficient, and the real-valued function $u(t)$ is a coherent control. The dissipative superoperator $\mathcal{L}_{n(t)}(\rho)$ describing the interaction with the environment is
\begin{equation*}
\mathcal{L}_{n(t)}(\rho) = n(t) \left(\sigma^+\rho\sigma^- + \sigma^-\rho_t\sigma^+ - \dfrac{1}{2}\{\sigma^-\sigma^+ + \sigma^+\sigma^- , \rho\}\right) + \left(\sigma^+\rho\sigma^- -  \dfrac{1}{2}\{\sigma^-\sigma^+, \rho\}\right).
\end{equation*}
Here, matrices $\sigma^\pm$ are
\begin{equation}
\sigma^+ = \begin{pmatrix}
0 & 1 \\
0 & 0
\end{pmatrix},\qquad\sigma^-= \begin{pmatrix}
0 & 0 \\
1 & 0
\end{pmatrix},
\end{equation}
and the non-negative function $n(t) \geq 0$ is an incoherent control. 

This master equation can describe the physical situation of an atom interacting with a generally time-dependent laser field and immersed in a time-dependent bath of incoherent photons, for example, with a bath of photons with a time-dependent temperature or with shaped incoherent light dynamically tailoring the spectrum of a broadband incoherent source to control the atomic and
molecular scale kinetics. In this case, incoherent control can physically be the spectral density of the (incoherent) environmental photons. This method was theoretically proposed in~\cite{Pechen_Rabitz_2006}. An experimental realization for the optimal control of the evolving hue in near-IR to visible up-converting phosphors, which mimics various aspects of chemical reaction kinetics including non-linear behavior, was implemented in~\cite{LaforgeJCP2018}. Another approach that allows for the experimental modification of decoherence rates was realized via the experimental control of the system--reservoir interaction and hydrogen spin relaxation rates~\cite{Pires2023}. Non-unitary control is also necessary for the experimental realization of all-to-one Kraus maps~\cite{Zhang_Saripalli_Leamer_Glasser_Bondar_2022}. 

For a two-level system, it is convenient to use Bloch ball parameterization of the density matrix $\rho$ using a vector $\br \in \R^3$, $\|\br\| \leq 1$:
\begin{equation}
\rho = \dfrac{1}{2} \Bigg(\mathbb{I} + \sum_{j = 1,2,3}r_j\sigma_j\Bigg), \qquad r_j = \mathrm{Tr}\rho \sigma_j,\qquad j = 1,2,3, 
\label{bloch_variables}
\end{equation}
where $\mathbb{I}$ is the identity matrix and $\sigma_i$ are Pauli matrices. The evolution of the Bloch vector  in this parameterization takes the following inhomogeneous form:
\begin{align}
\frac{d\mathbf{r}}{dt} &= A(u(t), n(t)) \mathbf{r} + \mathbf{b} = \left(B + B^u u(t) + B^n n(t) \right) \mathbf{r} + \mathbf{b}\nonumber  \\
&=
\begin{pmatrix}
-\gamma(n(t) + 1/2) & \omega & 0\\
-\omega & -\gamma(n(t) + 1/2) & -2\mu u(t)\\
0 & 2\mu u(t) & -2\gamma(n(t) + 1/2)
\end{pmatrix}
\br + 
\begin{pmatrix}
0\\
0\\
\gamma
\end{pmatrix}
,\quad \mathbf{r}(0) = \mathbf{r}_0,
\label{bloch_equation}
\end{align}
where $r_{0j} = \mathrm{Tr}\rho_0 \sigma_j$, $j = 1,2,3,$ and
\begin{multline*}
B = \begin{pmatrix}
- \dfrac{\gamma}{2} & \omega & 0 \\
- \omega & - \dfrac{\gamma}{2} & 0 \\
0 & 0 & - \gamma \\
\end{pmatrix},\, B^u = 2\mu
\begin{pmatrix}
0 & 0 & 0 \\
0 & 0 & -1 \\
0 & 1 & 0 \\
\end{pmatrix},\, B^n = -\gamma
\begin{pmatrix}
1 & 0 & 0 \\
0 & 1 & 0 \\
0 & 0 & 2 \\
\end{pmatrix},\,\mathbf{b} = 
\begin{pmatrix}
0 \\
0 \\
\gamma \\
\end{pmatrix}.
\end{multline*}

We also use the extension of the Bloch ball representation~(\ref{bloch_variables}) to a four-component vector using all four Hermitian $2\times 2$ basis matrices as was the case in~\cite{PetruhanovPechenJPA2023}:
\begin{equation}
\rho = \frac{1}{2}\sum_{i = 0}^4 q_i \sigma_i,\qquad i = 0,1,2,3,
\label{4_bloch_variables}
\end{equation}
where $\sigma_0 =\I$. Since $q_0 = {\rm Tr}\rho = 1$, the 0-coordinate is always equal to unity. In this representation, $\bq(t) = (1, \br(t))$ and the system~(\ref{bloch_equation}) takes a homogeneous form
\begin{equation}
\frac{d\mathbf{q}}{dt} = C(u(t), n(t))\bq = \left(\begin{array}{c|c}
0 & 0 \\
\hline
b & A(u(t), n(t))\end{array}\right)\bq, \qquad \bq(0) =  (1, \br_0).
\label{4_bloch_equation}
\end{equation}

Let us introduce the evolution operator $\Phi(t, u, n)$ (a dynamical map), which is a completely positive and a trace-preserving superoperator and is also called a quantum channel. It allows the solution of the system~(\ref{system_qubit}) to be written under control $(u, n)$ at time $t\in[0,T]$ as:
\begin{equation}
\rho(t, u, n) = \Phi(t, u, n) \rho_0.
\label{evolution_operator}
\end{equation}
Denote $\Psi(t, u, n)$ as a matrix representation of the operator $\Phi(t, u, n)$ in the basis $\sigma_i/2$~(\ref{4_bloch_variables}): $\bq(t, u, n) =  \Psi(t, u, n)\bq(0)$. The matrix $\Psi(t, u, n)$ is the solution of the system similar to~(\ref{4_bloch_equation}):
\begin{equation}
\der{\Psi}{t} = C(u(t), n(t))\Psi, \qquad \Psi(0) = \I.
\label{4_bloch_equation_evolution_matrix}
\end{equation}

The evolution operator $\Phi(t, u, n)$ is trace-preserving: $\Tr(\rho(T, u, n)) = \Tr(\Phi(t, u, n)\rho_0) = \Tr\rho_0 = 1$. This property, in terms of the matrix $\Psi(t, u, n)$, corresponds to preservation of the 0-coordinate: $q_0(t, u, n) = q_0(0) = 1$. Therefore, the matrix $\Psi(t, u, n)$ has the following form:
\begin{equation}
\Psi = 
\left(\begin{array}{c|c}
1 & 0 \\
\hline
\Psi' & \Psi''\end{array}\right),
\label{4_bloch_equation_evolution_matrix_structure}
\end{equation}
where $\Psi'$ is a $3\times1$ matrix and $\Psi''$ is a $3\times3$ matrix. Since $\bq = (1, \br)$, for the evolution of the vector $\br(t)$, we have (notations of controls are omitted for brevity):
\begin{equation}
\br(t) = \Psi''(t)\br_0 + \Psi'(t).
\label{bloch_vector_evolution_structure}
\end{equation}
Thus, in the Bloch ball representation, the evolution of the open quantum system is a composition of the linear map $\Psi''$ and translation by the vector $\Psi'$. 

If an evolution $\Phi$ is unital, i.e., preserving the maximally mixed state: $\Phi(\I/2) = \I/2$, then the matrix $\Psi$ in the Bloch parameterization preserves the $(0, 0, 0)$ state; therefore, its inhomogeneous part equals zero, $\Psi' = 0$, and the evolution is linear for the vector $\br(t)$: 
\begin{equation}
\br(t) = \Psi''(t)\br_0 .
\label{bloch_vector_evolution_structure_linear}
\end{equation}

\section{Objective Functionals for the Gate Generation Problem}\label{Sec3}

The problem of gate generation is to find a control $(u,n)$ that will produce a dynamic map $\Phi(T, u, n)$ (\ref{evolution_operator}), which is an evolution operator at time $t = T$, that coincides or is as close as possible to the desired unitary (or more generally, non-unitary) gate $U$ operation: $\Phi(T, u, n) = U\cdot U^\dagger$. %Please ensure the meaning has been retained. 
%Authors: we changed to retain the meaning.
 It can happen  that no admissible control gives an exact equality between  the actual evolution and the target operation. For example, in~\cite{Lokutsievskiy_2021}, it was shown that the set of attainability for the system~(\ref{system_qubit}) from the poles does not fill the whole Bloch ball. Particularly, Hadamard eigenstates $\ket{+}$ and $\ket{-}$ are not reachable from the poles for the distance $\sim\gamma/\omega$ due to the presence of decoherence. Therefore, the system~(\ref{system_qubit}) does not allow for the exact   generation of, e.g., the Hadamard gate, or any other gate besides the rotation around the $z$-axis in the Bloch ball. 

Thus, the gate generation problem should be formulated as to find a control $(u,n)$ that will produce a dynamic map $\Phi(T, u, n)$ that is as close as possible to the desired target (in our case, the unitary operation $U\cdot U^\dagger$) in the sense of some chosen metric or measure of distance. Then, the goal is to minimize a proper objective functional defined by $\Phi(T, u, n)$ and $U$. 
One of the possible options is to examine the action of the dynamic map $\Phi(T, u, n)$ on a chosen set of states or matrices in the qubit Hilbert space.

In this work, we considered four functionals. The first three are the objective functionals and have the form of a mean value of the squared Hilbert--Schmidt distance between actions of $\Phi(T, u, n)$ and $U \cdot U^\dagger$ on some set of density matrices $\{\rho_0^{(j)}\}_{j = 1}^K$:
\begin{equation}
F_{U,K}\left(u, n; \rho_0^{(1)}, \dotsc, \rho_0^{(K)}\right) = \frac{1}{K}\sum_{j = 1}^K \left\| \Phi(T, u, n) \rho_0^{(j)} - U\rho_0^{(j)}U^\dagger \right\|^2,
\label{functional_on_states}
\end{equation}
This class of objective functionals was also used in~\cite{PetruhanovPhotonics2022}. We consider three different sets with $K =$ 2, 3, and 4 matrices, which give the first $F_{U,2}$, the second $F_{U,3}$, and the third $F_{U,4}$ objectives, respectively.
\begin{itemize}

\item The first set $\{\rho_0^{(1)}, \rho_0^{(2)}\}$ corresponds to basis states $\ket{0}$ and $\ket{1}$ in $\mathcal{H} = \C^2$:
\begin{equation}
\rho_0^{(1)} = \ket{0}\bra{0},\qquad \rho_0^{(2)} = \ket{1}\bra{1}.
\label{first_set_basis_states}
\end{equation}

\item The second set corresponds to three states determining the implementation of the unitary operation among all dynamic maps~\cite{Goerz_NJP_2014_2021}:
\begin{equation}
\rho_0^{(1)} = \begin{pmatrix}
2/3 & 0\\
0 & 1/3
\end{pmatrix}, \qquad \rho_0^{(2)} = \begin{pmatrix}
1/2 & 1/2\\
1/2 & 1/2
\end{pmatrix}, \qquad \rho_0^{(3)} =
\begin{pmatrix}
1/2 & 0\\
0 & 1/2
\end{pmatrix}.
\label{second_set_basis_states}
\end{equation}
We sometimes call the objective functional defined using this set as the GRK (Goerz--Reich--Koch)-type objective functional. 

\item The third set corresponds to four basis Hermitian matrices in the linear space in which the dynamic maps act:
\begin{equation}
\begin{split}
\rho_0^{(1)} &= \ket{0}\bra{0} = 
\begin{pmatrix}
    1 & 0 \\
    0 & 0
\end{pmatrix}, \qquad 
\rho_0^{(2)} = \ket{1}\bra{1} = 
\begin{pmatrix}
    0 & 0 \\
    0 & 1
\end{pmatrix},\\ 
\rho_0^{(3)} &= \ket{+}\bra{+} = \frac{1}{2}
\begin{pmatrix}
    1 & 1 \\
    1 & 1
\end{pmatrix}, \qquad
\rho_0^{(4)} = \ket{i}\bra{i} = \frac{1}{2}
\begin{pmatrix}
    1 & -i \\
    i & 1
\end{pmatrix}.
\end{split}
\label{third_set_basis_states}
\end{equation}
\end{itemize}
A more detailed discussion of these three sets is provided in Appendix~\ref{Sec:AppendixA}. 

The fourth functional is used to analyze the implementation of the generated unitary gates and equals the squared Hilbert--Schmidt distance in the space of operators on $\C^{2\times2}$:
\begin{equation}
    F_U(u, n) = \|\Phi(T, u, n) - U \cdot U^\dagger\|^2.
    \label{Frobenius_norm_functional}
\end{equation}
Unlike the objective functionals $F_{U,2}$, $F_{U,3}$, and $F_{U,4}$, this functional is defined by the true distance and allows one to define the closeness of an evolution operator of the system to the desired unitary operation.

\section{Gradient-Based Optimization Method}\label{Sec4}

For the optimization of the functional~(\ref{functional_on_states}), we used our modification of the gradient-based GRAPE method, which was originally proposed for the generation of NMR sequences~\cite{Khaneja2005}. The gradient-based approach for controlling open quantum systems driven by \textbf{\emph{both coherent and incoherent controls}} was developed recently in~\cite{PetruhanovPechenJPA2023} for general $N$-level open quantum systems, where the exact solution for a qubit was found. The implementation of this method for the considered one-qubit quantum system was performed in~\cite{PetruhanovPhotonics2022}. Here, we briefly present the concept and provide the basic expression for the gradient of our objective functionals.

To solve the optimization problem of a unitary gate $U$ generation:
\begin{equation}
F_{U,K}\left(u, n; \rho_0^{(1)}, \dotsc, \rho_0^{(K)}\right) \to \inf_{u, n},
\label{optimization_problem}
\end{equation}
we used piecewise constant controls:
\begin{align}
u(t) &= \displaystyle \sum_{k=1}^M u_k \chi_{[t_{k-1}, t_k)}(t), \qquad u_k\in\mathbb R\label{PConst_qubit_control:u}\\
w(t) &= \displaystyle \sum_{k=1}^M w_k \chi_{[t_{k-1}, t_k)}(t), \qquad w_k\in\mathbb R\label{PConst_qubit_control:w}\\
n(t) &= w(t)^2,
\end{align}
where $0 < t_0 < t_1 < \dots < t_M = T,$ and $\chi_{[t_{k-1}, t_k)}$ is the characteristic function of the half-open interval $[t_{k-1}, t_k)$. The new control $w(t)$ is introduced in order to achieve the optimization of the space $\R^{2M}$ of control components $u_k$ and $w_k$ and to avoid dealing with the boundary $n(t)\geq 0$.
For brevity, we denote pair $(u,w)$ as $v = (v^1, v^2)$.
The use of piecewise constant control leads to a piecewise constant r.h.s of the evolution equation~(\ref{bloch_equation}) so that the evolution is given by a sequence of states $\br(t_k)$:
\begin{multline}
\mathbf{r}_k \equiv \mathbf{r}(t_{k}) = e^{A_k \Delta t_k} \mathbf{r}_{k - 1} + \mathbf{g}_k = e^{A_k \Delta t_k} \dotsm e^{A_1 \Delta t_1} \mathbf{r}_0 \\+ e^{A_k \Delta t_k} \dotsm e^{A_2 \Delta t_2} \mathbf{g}_1  +\dotsb + e^{A_k \Delta t_k} \mathbf{g}_{k-1} + \mathbf{g}_{k},
\label{bloch_k_th_state}
\end{multline}
where
\begin{equation}
\mathbf{g}_k = (e^{A_k \Delta t_k} - \mathbb{I}) A_k^{-1} \mathbf{b}, \quad \Delta t_k  = t_k - t_{k - 1}, \quad k = 1, \dots, M.
\label{g_definition}
\end{equation}

Denote Bloch vectors of $\Phi(T, u, n)\rho_0^{(j)}$ and $U\rho_0^{(j)}U^\dagger$ as $\mathbf{r}^{(j)}(T, v)$ and $\mathbf{r}^{(j)}_U$, respectively. For numeric optimization, it can be useful to use parametrization for the objective functional~(\ref{functional_on_states}), which we provide in Appendix~\ref{Sec:AppendixB}. The gradient of the objective functional~(\ref{functional_on_states}) with respect to controls~(\ref{PConst_qubit_control:u})~and~(\ref{PConst_qubit_control:w}) is~\cite{PetruhanovPhotonics2022}:
\begin{equation}
\pd{F_{U,K}}{v^m_k} = 
\frac{1}{K}\sum_{j = 1}^K \left(\mathbf{r}^{(j)}(T, v) -  \mathbf{r}^{(j)}_U\right) \cdot \pd{\mathbf{r}^{(j)}(T, v)}{v^m_k},\quad k = 1, \dotsc, M,\quad m = 1,2.
\label{functional_on_states_bloch_gradient}
\end{equation}
The gradient of the final state is given by the following expressions~\cite{PetruhanovPhotonics2022, PetruhanovPechenJPA2023}:
\begin{multline} 
\pd{\mathbf{r}(T)}{v^m_k}  = e^{A_N \Delta t_N} \dots e^{A_{k + 1} \Delta t_{k + 1}}\bigg[ \pd{}{v_k^m}\left(e^{A_k \Delta t_k}\right)\mathbf{r}_{k-1} + \pd {\mathbf{g}_k}{(u_k, w_k)}\bigg],\\
k = 1, \dots, M,\quad m = 1,2.
\label{final state gradient u and w}
\end{multline}
The explicit expressions for the derivatives $\pd{}{v_k^m}\left(e^{A_k \Delta t_k}\right)$ and $\pd {\mathbf{g}_k}{(u_k, w_k)}$ are provided in Appendix~\ref{Sec:AppendixC}.

Finally, to numerically optimize the objective functional, we used the gradient descent method. The $(i+1)$th iteration of the algorithm gives a control $v^{(i)}$, starting from an initial guess $v^{(0)}$:
\begin{equation}
v^{(i + 1)} = v^{(i)} - h^{(i)} \mathrm{grad}_{v} F_{U,K}\big(u^{(l)}, {w^{(l)}}^2; \rho_0^{(1)}, \dotsc, \rho_0^{(K)}\big),\quad i = 0,1,\dotsc.
\label{gradient_descent}
\end{equation}
For the step length $h^{(i)}$, we used the adaptive scheme proposed in~\cite{PetruhanovPhotonics2022}. As a stopping criterion, we used the standard stop criterion corresponding to the first-order optimality condition; iterations continued until
\begin{equation}
F_{U,K}\left(u^{(i)}, n^{(i)};\rho_0^{(1)}, \dotsc, \rho_0^{(K)}\right) < \varepsilon.
\label{stop_criterion}
\end{equation}

\section{Numeric Analysis of the Control Landscapes}\label{Sec5}

We studied the properties of the quantum control landscape for the generation of unitary $H$ and $T$ gates~(\ref{H_T_gates}) and (\ref{optimization_problem}), i.e., for minimizing the objective functionals~$F_{U,K}$ for sets of the two~(\ref{first_set_basis_states}), three~(\ref{second_set_basis_states}), and four~(\ref{third_set_basis_states}) matrices in the qubit Hilbert space. One of the important properties of quantum control landscapes is the local but not global minima, also called traps. To study the quantum control landscape of the objective functional~$F_{U,K}\left(u, n; \rho_0^{(1)}, \dotsc, \rho_0^{(K)}\right)$, we performed the following statistical experiment. For each $l = 1, \dots, L$, we generated a random initial guess $\left(u^{(0),l}, n^{(0),l}\right)$ uniformly distributed in the hyper-rectangle (orthotope) $([-1, 1]\times[0,1])^{\times M}$:
\begin{equation*}
u^{(0),l} \in [-1, 1], \qquad n^{(0),l} \in [0,1], \qquad l = 1, \dots, L.
\end{equation*}
Then, using the gradient descent method described in Section~\ref{Sec4}, we found an optimized control $\left(u^{*,l}, n^{*,l}\right)$. The obtained values of the functional~$F_{U,K}\left(u^{*,l}, n^{*,l}; \rho_0^{(1)}, \dotsc, \rho_0^{(K)}\right)$ (in the figures, they are denoted as $F_{U,K}$, $U \in (H, T)$, $K = 2, 3, 4$) belong to a certain distribution on the range of all possible values of the objective functional. We built histograms approximating this distribution for each gate and each objective.  For comparison, we also built a distribution of the obtained values of the objective functional~(\ref{Frobenius_norm_functional}) $F_U\left(u^{*,l}, n^{*,l}\right)$ (in the figures, they are denoted as $F_{U}$, $U \in (H, T)$, $K = 2, 3, 4$) based on the Frobenius norm between the optimized evolution and the desired unitary gate operation. 

Numerical simulations were performed by writing a Python program using the Numpy library for fast matrix operations, the SciPy library functions \texttt{scipy.linalg.expm} for matrix exponential computation using the Pad\'{e} approximation, and \texttt{scipy.linalg.inv} for matrix' inverse computation. The following values of the system parameters were used: the transition frequency $\omega = 1$, the~dipole moment $\mu = 0.1$, the decoherence rate coefficient $\gamma = 0.01$, and the regular partition of the time segment~$[0, T]$ with $T = 5$ into $M = 10$ segments, such that each segment has the length $\Delta t_k = T/M = 0.5$. The stopping parameter $\varepsilon = 10^{-5}$, the parameters of the adaptive scheme of the step length~\cite{PetruhanovPhotonics2022} were $h^{(0)} = 1$, $c = 1.1$, $d = 0.5$, $L_{\rm stuck} = 20$. Integral formulae were approximated using the trapezoidal rule with the number of partitions $N_{\rm partition} = 20$. 

For gates $H$ and $T$, in Figure~\ref{Fig1:H} and Figure~\ref{Fig2:T}, respectively, for each of the three objective functionals, we built histograms of the distributions of best objective values (minimal infidelities) obtained with a gradient search starting from $1000$ various randomly generated conditions in the hyper-rectangle initial conditions. 
 These histograms approximate the distributions of best objective values. 

The obtained centers and widths of each peak, which were computed as the mean value and doubled standard deviation for each peak, are presented in Table~\ref{Tab1}.
\begin{table}[ht!]
\centering
\caption{Centers ($C_1$, $C_2$) and widths ($W_1,W_2$) of the obtained distributions for objective functionals defined with number of matrices $j=2,3,4$. For cases with two peaks, the first is the peak shown as green in Fig.~\ref{Fig3:T}.} 
\begin{tabular}{ |c|c|c c|c c| } 
\hline
\multicolumn{2}{|c|}{Func.} & $C_1$ & $W_1$ & $C_2$ & $W_2$ \\
\hline
\multirow{2}{2.5em}{$j = 2$} & $F_{H,2}$ & $1.601\times10^{-3}$ & $4.227\times10^{-5}$ & - & - \\ 
\cline{2-6}
                             & $F_{H}$   & $5.611$             & $4.627\times10^{-1}$ & - & -  \\ 
\hline
\multirow{2}{2.5em}{$j = 3$} & $F_{H,3}$ & $3.484\times10^{-4}$ & $1.276\times10^{-5}$ & - & -  \\ 
\cline{2-6}
                             & $F_{H}$   & $2.657\times10^{-2}$ & $1.979\times10^{-3}$ & - & -  \\ 
\hline
\multirow{2}{2.5em}{$j = 4$} & $F_{H,4}$ & $7.525\times10^{-4}$ & $2.317\times10^{-5}$ & - & -  \\ 
\cline{2-6}
                             & $F_{H}$   & $5.183\times10^{-3}$ & $1.527\times10^{-4}$ & - & -  \\ 
\hline
\multirow{2}{2.5em}{$j = 2$} & $F_{T,2}$ & $2.374\times10^{-3}$ & $1.236\times10^{-5}$ & - & - \\ 
\cline{2-6}
                             & $F_{T}$   & $4.795\times10^{-1}$ & $1.346\times10^{-2}$ & - & - \\ 
\hline
\multirow{2}{2.5em}{$j = 3$} & $F_{T,3}$ & $5.964\times10^{-4}$ & $6.720\times10^{-6}$ & $9.495\times10^{-4}$ & $1.821\times10^{-5}$\\ 
\cline{2-6}
                             & $F_{T}$    & $1.111\times10^{-2}$ & $6.866\times10^{-4}$ & $1.091\times10^{-2}$ & $3.094\times10^{-4}$ \\ 
\hline
\multirow{2}{2.5em}{$j = 4$} & $F_{T,4}$ & $1.317\times10^{-3}$ & $1.592\times10^{-5}$ & $1.624\times10^{-3}$ & $1.998\times10^{-5}$ \\ 
\cline{2-6}
                             & $F_{T}$ & $6.718\times10^{-3}$ & $2.592\times10^{-5}$ & $6.599\times10^{-3}$ & $2.735\times10^{-5}$\\ 
\hline
\end{tabular}
\normalsize
\label{Tab1}
\end{table}

\begin{figure}[ht!]
%\centering
\includegraphics[width = 1\linewidth]{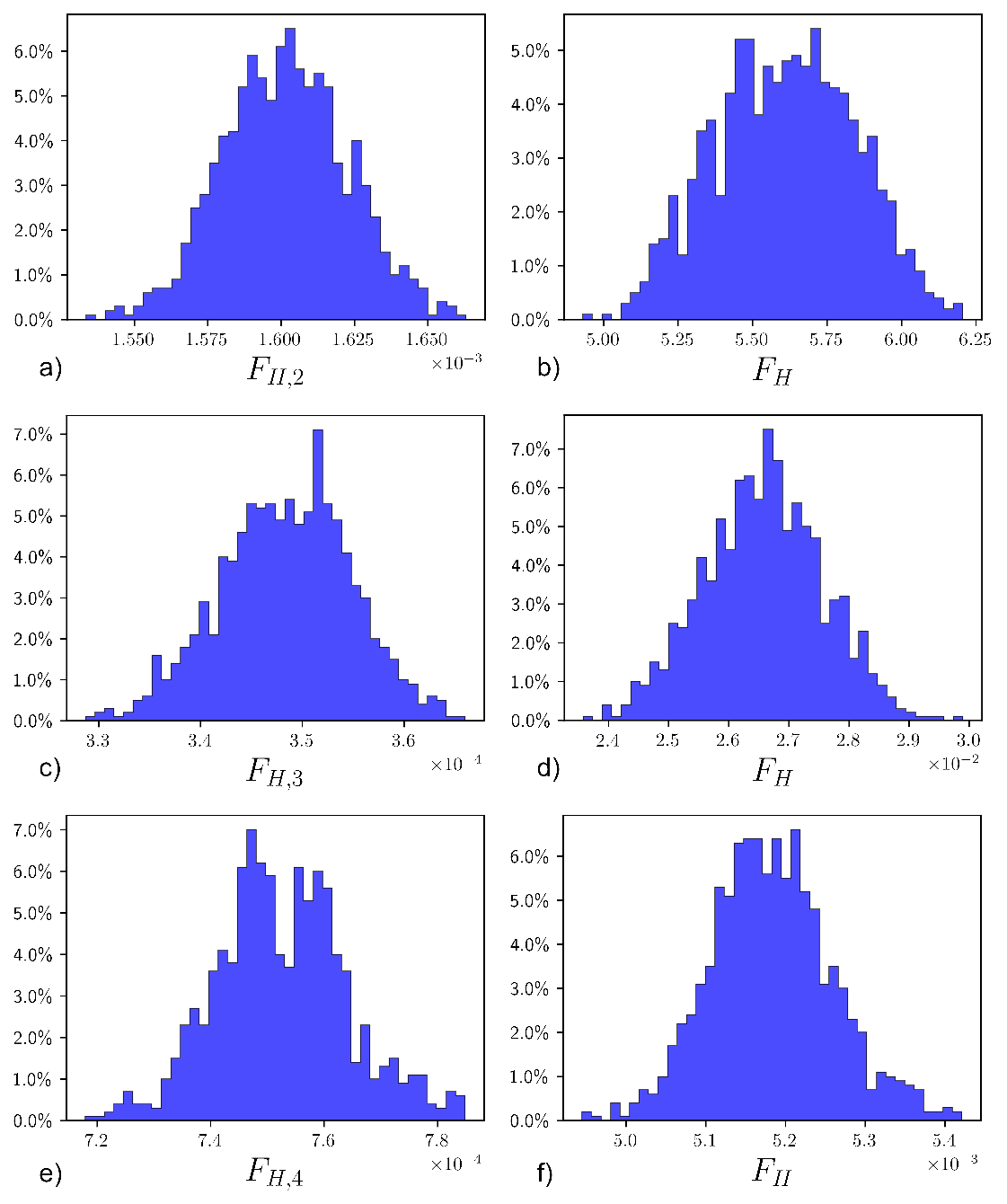}
\caption{Histograms describing distributions obtained by GRAPE values of the objective functionals for generating the $H$ gate. The values obtained starting from $L=1000$ random initial conditions uniformly distributed in some hyper-rectangle. Left column: for two (\textbf{a}), three (\textbf{c}), and four (\textbf{e}) matrices. Right column: the Frobenius norm $F_H$ with optimized controls for two (\textbf{b}), three (\textbf{d}), and four (\textbf{f})~matrices.\label{Fig1:H}}
\end{figure}

The results are surprising. For the Hadamard gate, the distributions of the best objective values obtained with the gradient search  for all considered objectives have a simple form with just one peak. However, for the $T$ gate, the situation is completely different---the distribution for the objective functional defined by two matrices also has one peak, whereas distributions of the best obtained values for the objective functionals defined by three and four matrices have two isolated peaks. This might indicate the possible presence of two isolated minima with different values of the infidelity for these two objective functionals. It is important that, as discussed above, among these three objective functionals, the smallness of the objective with two matrices does not guarantee the closeness of the generated process to the target gate. Only objectives defined with three and four matrices guarantee the closeness of the generated gate to a target and can be used as a good measure for the gate generation problem, and exactly for these objectives, we observe two peaks. 

\begin{figure}[ht!]
%\centering
\includegraphics[width = 1\linewidth]{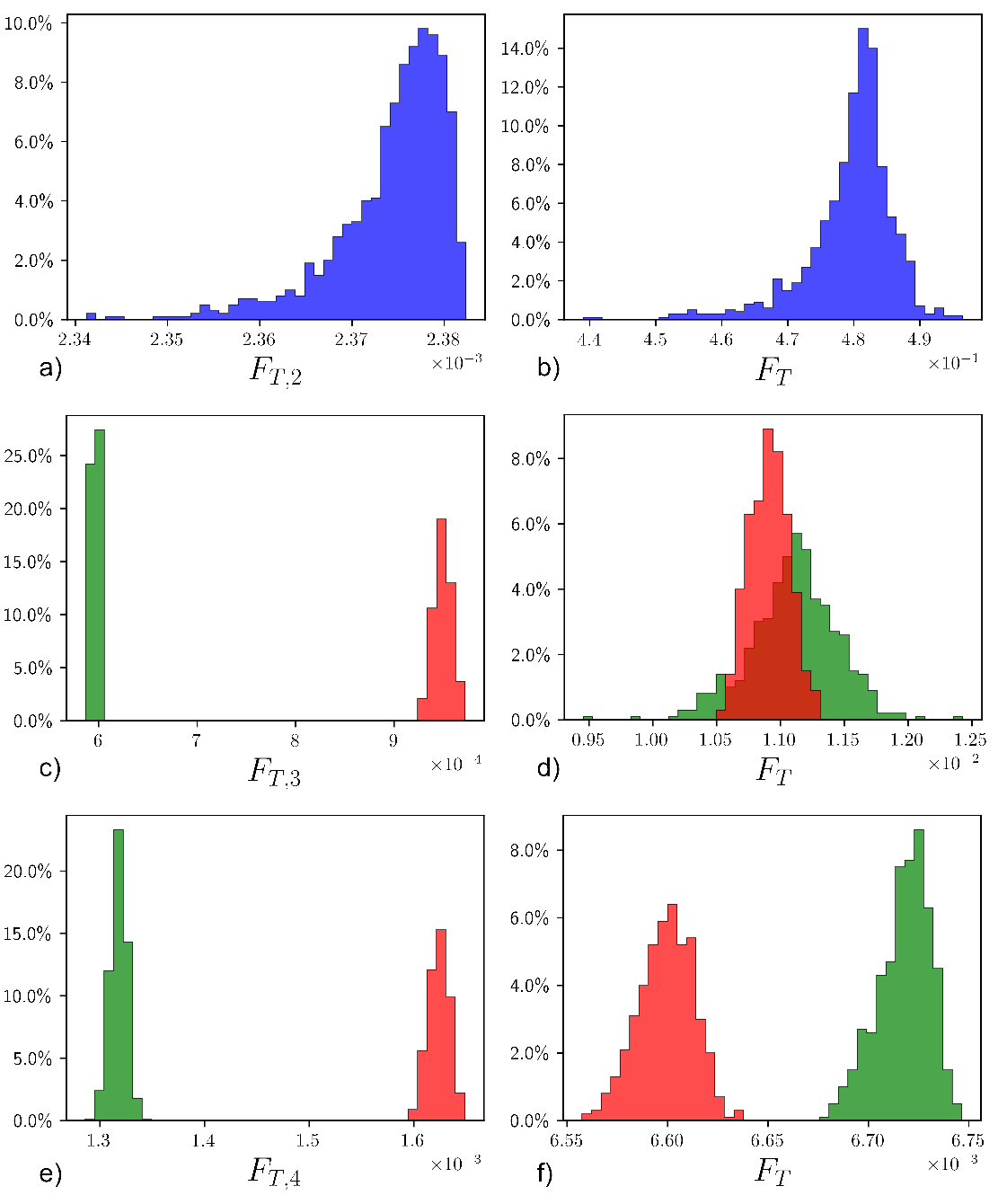}
\caption{Histograms describing distributions obtained by GRAPE values of the objective functionals for generating the $T$ gate. The values obtained starting from $L=1000$ random initial conditions uniformly distributed in some hyper-rectangle. Left column: for two (\textbf{a}), three (\textbf{c}), and four (\textbf{e}) matrices. Right column: the Frobenius norm $F_T$ with optimized controls for two (\textbf{b}), three (\textbf{d}), and four (\textbf{f}) matrices. Two separate peaks are shown in green and red colors.\label{Fig2:T}}
\end{figure}

To study cases with two peaks if these peaks correspond to two different groups of controls, we analyzed the optimized controls used for sub-plots (c) and (e) in Figure~\ref{Fig2:T}  in more detail. For each of these two sub-plots, we plotted all 1000 optimized coherent and incoherent controls for the generation of the $T$ gate in Figure~\ref{Fig3:T}, i.e., left and right sub-plots. The upper row corresponds to sub-plot (c) in Figure~\ref{Fig2:T}, and the bottom row to sub-plot (e). The coherent controls are clearly divided in two groups, which are symmetric with respect to $u=0$ line and are separated in the functional space of controls. Incoherent controls are also divided in two distinct sub-manifolds, which are separated in the space of controls, although the controls in these subgroups do intersect. Solutions to quantum optimal control problems are known to not comprise isolated points, but form multi-dimensional sub-manifolds in the control space, as was shown for finite-dimensional~\cite{RabitzPRA2006} and continuous-variable systems~\cite{LaroccaPRA2020}. Our finding confirms these results for the most general and previously unexplored case when both coherent and incoherent controls are used. %Please ensure the meaning has been retained. 
%we confirm

\begin{figure}[ht!]
%\centering
\includegraphics[width = \linewidth]{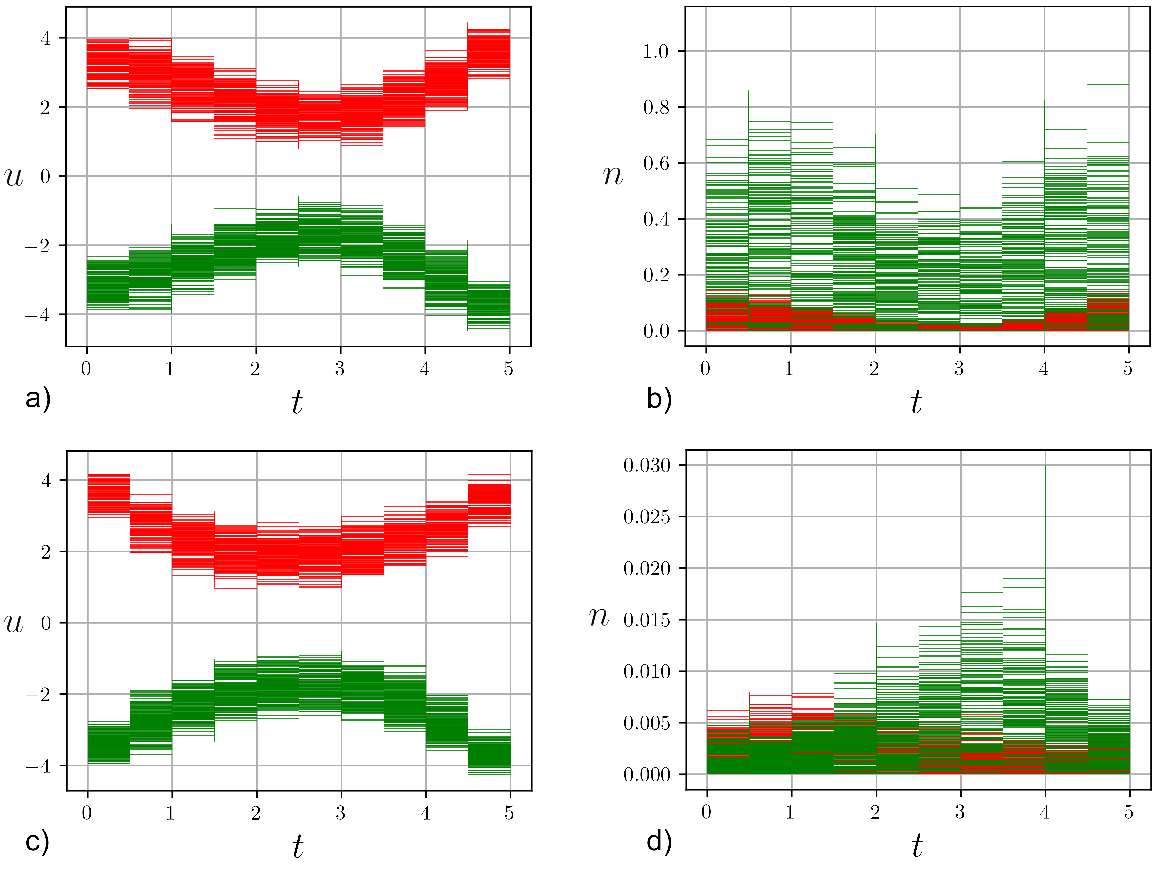}
\caption{Plots of all 1000 coherent (\textbf{a},\textbf{c}) and incoherent (\textbf{b},\textbf{d}) controls obtained by GRAPE optimization for generation of the $T$ gate. Sub-plots (\textbf{a}) and ({\bf b}): for the objective functional $F_{T,3}$. Sub-plots (\textbf{c}) and ({\bf d}): for the objective functional $F_{T,4}$.  Green (resp., red) color shows all controls leading to the left (resp., right) peak on the corresponding histograms in Figure~\ref{Fig2:T}. Both coherent and incoherent controls are clearly combined in two groups.}
\label{Fig3:T}
\end{figure}

\section{Discussion and Open Problems}\label{Conclusions}

We studied several control landscapes for the problem of the optimal generation of single-qubit $H$ and $T$ gates using coherent control and the environment as a constructive resource acting on the qubit via incoherent control. We considered three different landscapes for this problem, corresponding to the three different objective functionals defined by steering between two, three, and four matrices in two-dimensional qubit Hilbert space. We built histograms of distributions of the best objective values (minimal infidelities) obtained with a gradient search starting from $1000$ various randomly generated conditions in some hyper-rectangle initial conditions. These histograms approximate the best obtained distributions from the GRAPE objective values. 

The following observations were made from the obtained results.

First, for the Hadamard gate, which is a Clifford gate, all the considered objective distributions of the best objective values obtained with the gradient search have a simple form with just a one peak. However, for the $T$ gate, which is a non-Clifford gate, the situation is completely different---the distribution of the best values for some objective functionals has one peak, whereas the distribution of the best values for the other objectives has two isolated peaks, which might indicate the possible existence of two isolated minima in the control landscape with different values of the fidelity. This finding, and the question as to whether there is a relationship between the fact that the $H$ gate is a Clifford gate and the $T$ gate is not, requires further analysis. 

Second, a less obvious finding is that, for the $T$ gate, the distribution of the GRAPE-minimized GRK-type objective $F_{T,3}$ with three matrices has two significantly separated peaks, whereas the distribution of the objective computed with the Frobenious norm for the same optimized controls has two overlapping peaks corresponding to two separate groups of controls. Moreover, a larger value peak in the $F_{T,3}$ histogram has a lower value center in the Frobenius norm $F_{T}$ histogram. The same is true for the $F_{T,4}$ and the $F_{T}$ histograms.

Third, the minimal infidelities obtained for the objective functional $F_{U,2}$  defined with only two states are higher by one order of magnitude than the infidelities obtained for objective functionals $F_{U,3}$  and $F_{U,4}$ defined with three and four states. Naively, one could expect the opposite situation, since the objective functional $F_{U,2}$  may even be minimized by gates that act as the target only on the two basis states, but which can act differently from the target gate in some other states. Moreover, the objective $F_{U,4}$ includes $1/2F_{U,2}$ as a summand together with some non-negative term:
\begin{multline*}
    F_{U,4}\left(u, n; \rho_0^{(1)}, \rho_0^{(2)}, \rho_0^{(3)}, \rho_0^{(4)}\right)=\frac{1}{2}F_{U,2}\left(u, n; \rho_0^{(1)}, \rho_0^{(2)}\right) +\\ \frac14\biggl(\left\| \Phi(T, u, n) (\ket{+}\bra{+}) - U\ket{+}\bra{+}U^\dagger \right\|^2 + \left\| \Phi(T, u, n) (\ket{-}\bra{-}) - U\ket{-}\bra{-}U^\dagger \right\|^2\biggr),
\end{multline*}
where $\left\{\rho_0^{(1)}, \rho_0^{(2)}, \rho_0^{(3)}, \rho_0^{(4)}\right\}$ is the third set of states~(\ref{third_set_basis_states}).
However, if we optimize $F_{U,3}$ or $F_{U,4}$ and substitute the obtained controls in $F_{U,2}$, then we obtain smaller values of $F_{U,2}$ than if we optimize $F_{U,2}$ directly. This finding is counterintuitive and may indicate that $F_{U,3}$ and $F_{U,4}$ not only more exactly represent the problem of gate generation, but that they are also more appropriate for {\it efficient} GRAPE optimization. 

Fourth, in cases with two peaks, the controls corresponding to each of the peaks form subsets (sub-manifolds) in the space of controls. The fact that optimal controls form sub-manifolds in the space of controls was shown for the coherent control of finite-level and continuous systems in~\cite{RabitzPRA2006,LaroccaPRA2020}. Here, we discovered this feature for the most general situation when coherent and incoherent controls are used together, and moreover, we observed the existence of two such sub-manifolds separated by a large distance.

The physical meaning of the results is that an experimental local (e.g., gradient-based) search for optimal controls implementing a single-qubit $T$ gate may converge to two maxima with different infidelities: one of which is higher and, therefore, worse; whereas for the $H$ gate, we did not observe this behavior. This circumstance should be taken into account in such experiments on quantum controls. Another observation is that objectives defined using three or four states may be more suitable for optimization, as they provide lower values of the infidelities compared to the objective defined using two states (which may also give a non-unique solution). There is also a distinction between the behavior of minimal infidelities of these objectives and the objective defined by the Frobenious distance between the actual and the target maps---the latter, to some degree, ``hides'' the two local peaks found for the former. The finding of two sub-manifolds of (probably locally) optimal controls with different infidelities implies the importance of a suitable choice of the initial controls for optimization. An incorrect choice may lead to the convergence of the algorithm to a control without the best infidelity. All of these factors have to be taken into account for the practical experimental realization of the optimal generation of $H$ and $T$ gates.

\bigskip
\bigskip

\noindent {\bf Abbreviations:} 
\begin{itemize}
\item GKSL --- Gorini--Kossakowski--Sudarshan--Lindblad;
% \vspace{-0.3cm}
\item GRAPE --- GRadient Ascent Pulse Engineering;
% \vspace{-0.3cm}
\item GRK type objective --- objective functional for generating unitary gates under dissipative evolution where only the three initial density matrices as considered by M.Y.~Goerz, D.M.~Reich, and C.P.~Koch in~\cite{Goerz_NJP_2014_2021, Goerz_2021}. 
\end{itemize}
\medskip

\noindent {\bf Acknowledgement.}
This work was funded by the Ministry of Science and Higher Education of Russian Federation (Project No. 075-15-2020-788).
\medskip

\appendix
\appendixpage

\section{Sets of States}\label{Sec:AppendixA}

\subsection{First Set} 

The first set~(\ref{first_set_basis_states}) corresponds to two pure states $\ket{0}$ and $\ket{1}$ of the standard basis in $\mathcal{H} = \C^2$. These two states are sufficient to discriminate any operators on $\C^2$ by their action on the states. However, states of the quantum system are defined up to an arbitrary phase. Therefore, there can be a situation when different unitary gates act identically on these~states. 

In~\cite{PetruhanovPhotonics2022}, an example was provided when two different gates, the Hadamard gate~$H$ and the quantum gate $\sqrt{Y} \sim \exp(-i\pi/4\sigma_y)$ (rotation along y-axis by $\pi/2$ angle in the Bloch ball),  act identically on the initial states $\ket{0}\bra{0}$ and $\ket{1}\bra{1}$. Actually, this example can be generalized to all unitary operations that act on $\ket{0}\bra{0}$ and $\ket{1}\bra{1}$ identically as the Hadamard gate:
\begin{equation}
H\ket{0}\bra{0} H = \ket{+}\bra{+},\qquad H\ket{1}\bra{1} H = \ket{-}\bra{-},
\label{Hadamard_action_two}
\end{equation}
are rotations along axis $\mathbf{n} = \left(\cos\theta/\sqrt{2}, \sin\theta, \cos\theta/\sqrt{2}\right)$ by angle $2\arctan(1/\sin\theta)$, $\theta \in (-\pi/2, \pi/2]$, i.e., unitary quantum gates of the form
\begin{equation}
U(\theta) = \exp\left[-i \arctan\left(\frac1{\sin\theta}\right) \left(\frac{\cos\theta}{\sqrt{2}}\sigma_x + \sin\theta\sigma_y + \frac{\cos\theta}{\sqrt{2}}\sigma_z\right)\right], \quad\theta \in (-\pi/2, \pi/2].
\label{rotations_as_H}
\end{equation}
The quantum gates $H$ and $\sqrt{Y}$ are particular cases corresponding to $\theta = 0$ and $\pi/2$, respectively.

This non-uniqueness of unitary operations comes from the fact that the states $\ket{0}$ and $\ket{1}$ are orthogonal, i.e., their Bloch vectors are opposite. Indeed, in the Bloch ball, a unitary operation $U$ is rotation $\phi_U$ and, therefore, a linear map. If states are opposite, $\br$ and $-\br$, then knowing the action of the rotation $\phi_U$ on them only gives information about one action $\phi_U(\br)$ because of the linear dependence of the vectors. There are many rotations that act the same: they are all along the axis orthogonal to  $\phi_U(\br) - \br$ at the corresponding angle if $\phi_U(\br) \neq \br$, or along the axis $\br$ at any angle if $\phi_U(\br) = \br$. %Please ensure the meaning has been retained.
%we confirm

According to \cite{Reich2013}, there are two states $\{\rho_1, \rho_2\}$ that can distinguish any two unitary operations. This is equivalent to them having the commutant space $\mathcal{K}(\{\rho_1, \rho_2\})$ containing only the identity $\I$. It can be satisfied if $\{\rho_1, \rho_2\}$ is a complete and totally rotating set (see the definitions in \cite{Reich2013}). This finding corresponds to the statement mentioned above. If we take two orthogonal pure qubit states $\ket{\psi}$ and $\ket{\varphi}$, $\scalarproduct{\psi}{\varphi} = 0$, then the commutant space of their density matrices contains more than just the identity. For instance, the commutant space of $\ket{0}\bra{0}$ and $\ket{1}\bra{1}$ contains all the phase-shift gates $U_\delta$ including the identity:
\begin{equation}
\mathcal{K}(\{\ket{0}\bra{0}, \ket{1}\bra{1}\}) = \{U_\delta, \delta \in [0, 2\pi)\},\qquad U_\delta = \begin{pmatrix}
1 & 0\\
0 & e^{i\delta}
\end{pmatrix}.
\label{phaseshift_gate}
\end{equation}
Therefore, they cannot distinguish two unitary gates. 

Conversely, any two non-orthogonal qubit pure states $\ket{\psi}$ and $\ket{\varphi}$, $\scalarproduct{\psi}{\varphi} \neq 0$, form a complete and totally rotating set $\{\ket{\psi}\bra{\psi}, \ket{\varphi}\bra{\varphi}\}$. Completeness is obvious and the set is totally rotated with respect to, e.g., the state $\frac12\ket{\psi}\bra{\psi} + \frac12\ket{\varphi}\bra{\varphi}$. Thus, any two non-orthogonal qubit pure states $\ket{\psi}$ and $\ket{\varphi}$ can distinguish two unitary operations. Recalling that two non-orthogonal qubit pure states correspond to two non-opposite vectors in the Bloch ball and unitary operations correspond to rotations in the Bloch ball, this statement means that two non-opposite unit vectors are necessary and sufficient to determine a unique rotation of the Bloch ball, which is a simple geometric statement to be proved. It can be generalized to any two non-parallel (not necessarily unit length) vectors. 

Since we considered the dynamics of the open quantum system~(\ref{system_qubit}), a suitable set of states should distinguish not only two unitary operations, but also a unitary operation and a non-unitary operation. For this, three and more states might be used.

\subsection{Second Set} 
The second set~(\ref{second_set_basis_states}) corresponds to three special mixed states which, according to~\cite{Goerz_NJP_2014_2021}, are sufficient for the proper implementation of a unitary quantum gate in any $N$-level open quantum system. One cannot guarantee that the objective functional~(\ref{functional_on_states}) will be minimized exactly to zero as a result of the optimization. It might be that the optimized evolution will act on the states not exactly as the desired unitary operation, but with a small error. Moreover, it is not evident that generated evolution will differ from the desired unitary operation with an error of the same order. Therefore, the optimized evolution should be post-checked by, e.g., the functional~(\ref{Frobenius_norm_functional}), as it is shown in Section~\ref{Sec5}.

\subsection{Third Set}

The third set~(\ref{third_set_basis_states}) corresponds to the basis used in~\cite{PetruhanovPhotonics2022}: $\{\rho_0^{(j)}\}_{j=1}^4$ of $2\times2$ Hermitian matrices in the real four-dimensional linear space. %Please ensure the meaning has been retained. 
%we confirm
 For each unitary gate $U$, there exists a unique linear map that acts on these matrices as $U \cdot U^\dagger$. Moreover, since these states form the basis, a small value of the optimized functional~(\ref{functional_on_states}) means that the functional based on the Frobenius norm~(\ref{Frobenius_norm_functional}) also has a small value and the optimized evolution is close to the desired unitary operation. 

\section{Parametrization and Property of the Functionals}\label{Sec:AppendixB}

Here, we consider the parametrization~(\ref{4_bloch_variables}) of the functionals~(\ref{functional_on_states})~and~(\ref{Frobenius_norm_functional}) in the basis $M_k = \sigma_k/2$, $k = 0, 1, 2, 3$. The norm of any density matrix $\rho$ is equal to
\begin{equation}
\|\rho\|^2 = 1/2\|x\|^2.
\label{norm_relation}
\end{equation}
Taking into account~(\ref{norm_relation}), the objective functional~(\ref{functional_on_states}) becomes equal to
\begin{align}
F_{U,K}\left(u, n; \rho_0^{(1)}, \dotsc, \rho_0^{(K)}\right) &= \frac{1}{2K}\sum_{j = 1}^K \left\|(\Psi(T, u, n) - \Psi_U)(1, \br_0^{(j)}) \right\|^2\nonumber \\
&= \frac{1}{2K}\sum_{j = 1}^K \left\|\br^{(j)}(T, u, n) - \br_U^{(j)}\right\|^2.
\label{functional_on_states_bloch}
\end{align}
The basis $M_k = \sigma_k/2$ is orthonormal up to a constant: $\frac14{\rm Tr}(\sigma_i\sigma_j) = \frac12\delta_{ij},$ $i,j \in \{0, 1, 2, 3\}$. Therefore for the functional~(\ref{Frobenius_norm_functional}), the following invariance holds:
\begin{equation}
    F_U(u, n) = \|\Phi(T, u, n) - U \cdot U^\dagger\|^2 = \|\Psi(T,u,n) - \Psi_U\|^2,
    \label{functional_Frobenius_invariance}
\end{equation}
where $\Psi(t, u, n)$ and $\Psi_U$ are matrices of the evolution operator $\Phi(t, u, n)$ and the unitary operation $U \cdot U^\dagger$, respectively, in the basis $M_k = \sigma_k/2$. 

An interesting question is how the two different functionals (\ref{functional_on_states}) and (\ref{Frobenius_norm_functional}) are related. In particular, if, for the functional defined with states (\ref{functional_on_states}), we take the three standard basis states in $\R^3$:
\begin{equation}
\br_0^{(1)} = (1, 0, 0),\qquad\br_0^{(2)} = (0, 1, 0),\qquad\br_0^{(3)} = (0, 0, 1).
\label{standad_basis_R3}
\end{equation}

\begin{proposition}
If the evolution $\Phi(t, u, n)$ is unital and states $\rho_0^{(j)}$ correspond to~(\ref{standad_basis_R3}), then
\begin{equation}
F_U(u, n) = 6F_{U,3}\left(u,n;\rho_0^{(1)},\rho_0^{(2)},\rho_0^{(3)}\right).
\label{if_unital_equation}
\end{equation}
\end{proposition}
\begin{proof}
Denote $\bq^{(j)}$, $j = 0, 1, 2, 3$, as elements of the standard basis in $\R^4$. If the evolution is unital, i.e., $\Psi(t, u, n)$ preserves $\bq^{(0)} = (1, 0, 0, 0)$, then, given that unitary evolution is always unital,
\begin{equation}
(\Psi(t, u, n) - \Psi_U) \bq^{(0)} = \bq^{(0)} - \bq^{(0)} = 0.
\label{unitality_cosequence}
\end{equation}
The basis $\sigma_j/2$ is orthogonal up to a constant $c = 1/2$; therefore, according to~(\ref{norm_relation}),
\begin{equation}
\|(\Phi(t, u, n) - U \cdot U^\dagger)\rho^{(j)}\|^2 = 1/2\|(\Psi(t, u, n) - \Psi_U)(1, \br^{(j)})\|^2.
\label{norm_relation_for_proof}
\end{equation}
Thus, taking into account~(\ref{unitality_cosequence})~and~(\ref{norm_relation_for_proof}), we have
\begin{multline*}
F_U(u, n) = \|\Psi(t, u, n) - \Psi_U\|^2 = \sum_{j = 0}^3\|(\Psi(t, u, n) - \Psi_U)\bq^{(j)}\|^2 \\= \sum_{j = 1}^3\|(\Psi(t, u, n) - \Psi_U)\bq^{(j)})\|^2 = \sum_{j = 1}^3\|(\Psi(t, u, n) - \Psi_U)(\bq^{(j)} + \bq^{(0)})\|^2 \\= \sum_{j = 1}^3\|(\Psi(t, u, n) - \Psi_U)(1, \br_0^{(j)})\|^2 = 2\sum_{j = 1}^3\|(\Phi(t, u, n) - U \cdot U^\dagger)\rho_0^{(j)}\|^2 \\= 6F_{U,3}\left(u,n;\rho_0^{(1)},\rho_0^{(2)},\rho_0^{(3)}\right).
\end{multline*}
\end{proof}

\section{Gradient-Based Optimization Method}\label{Sec:AppendixC}

Here, we provide explicit expressions for derivatives used in Section~\ref{Sec4}:
\begin{align}
&\pd {\mathbf{g}_k}{v^1_k} = \pd {\mathbf{g}_k}{u_k} = \left(\pd{}{u_k}e^{A_k \Delta t_k} - (e^{A_k \Delta t_k} - \I)A_k^{-1} B^u\right) A_k^{-1} \mathbf{b},\label{g gradient u}\\ 
&\pd{}{v^1_k}e^{A_k \Delta t_k} = \pd{}{u_k}e^{A_k \Delta t_k} = \int \limits _0^{\Delta t_k} e^{A_k t }\, B^u\, e^{A_k (\Delta t_k-t)} {\mathrm{d}} t,\label{F gradient u}\\
&\pd {\mathbf{g}_k}{v^2_k} = \pd {\mathbf{g}_k}{w_k} = \left(\pd{}{w_k}e^{A_k \Delta t_k} - 2 w_k (e^{A_k \Delta t_k} - \I)A_k^{-1} B^n\right) A_k^{-1} \mathbf{b},\label{g gradient w}\\ 
&\pd{}{v^2_k}e^{A_k \Delta t_k} = \pd{}{w_k}e^{A_k \Delta t_k} = 2 w_k \int \limits _0^{\Delta t_k} e^{A_k t }\, B^n\, e^{A_k (\Delta t_k-t)} {\mathrm{d}} t,\quad k = 1, \dotsc, M.\label{F gradient w}
\end{align} 
The expressions~(\ref{F gradient u}) and~(\ref{F gradient w}) are calculated using the integral formula~\cite{Wilcox}:
\begin{equation}
\der{}{x} e^{A(x)} = \int \limits _0^1 e^{s A(x)} \der{A(x)}{x} e^{(1-s)A(x)} {\mathrm{d}} s.
\label{special_formula}
\end{equation}

\bibliographystyle{unsrturl}
\bibliography{refs.bib}

\end{document}